\documentclass[12pt,a4paper]{article}
\usepackage{tgpagella}
\usepackage[utf8]{inputenc}
\usepackage[top=1in, bottom=1in, left=1in, right=1in]{geometry}
\usepackage{mathtools, amsfonts, amssymb, amsthm, graphicx, xcolor, dsfont, hyphenat} % dsfont -- indicator f'n; hyphenat - wrapping hypernated words;
\usepackage[square, sort, comma, longnamesfirst]{natbib}
\usepackage[toc, page, title]{appendix}
\usepackage[colorlinks=false, citebordercolor = cyan, linkbordercolor = cyan]{hyperref}
\usepackage[labelfont=bf]{caption}
\hypersetup{linkbordercolor=cyan}
\def\spacingset#1{\renewcommand{\baselinestretch}%
{#1}\small\normalsize} \spacingset{1}
\DeclareMathOperator*{\E}{\mathbb{E}}

\DeclareMathOperator*{\argmin}{arg\,min}

\newcommand{\abs}[1]{\left\lvert#1\right\rvert}
\newcommand{\norm}[1]{\left\lVert#1\right\rVert}
\newcommand{\inner}[1]{\left\langle#1\right\rangle}
\newtheorem{theorem}{Theorem}[section]
\newtheorem*{theorem*}{Theorem}
\newtheorem{proposition}{Proposition}[section]
\newtheorem{lemma}{Lemma}[section]

\usepackage{neuralnetwork}
\newcommand*{\Scale}[2][4]{\scalebox{#1}{$#2$}}
\usepackage[ruled,vlined]{algorithm2e}
\usepackage{setspace} % "setstretch" algorithm
\usepackage{booktabs, longtable}

\title{The Kernel Trick for Nonlinear Factor Modeling\thanks{\scriptsize 
The paper has benefited from the suggestions and insightful comments received from the participants of the $40$th International Symposium on Forecasting, Econometric Society Winter Meeting, 2020, and the econometrics seminar at the University of California Riverside. This research did not receive any specific grant from funding agencies in the public, commercial, or not-for-profit sectors. Declarations of interest: none.}}
\author{Varlam Kutateladze\thanks{\scriptsize{Correspondence to: Department of Economics, University of California, Riverside, CA 92521, USA. \newline E-mail: \texttt{varlam.kutateladze@email.ucr.edu}.}}}
\date{\today}

\begin{document}
\maketitle
\begin{abstract}\small{
	Factor modeling is a powerful statistical technique that permits to capture the common dynamics in a large panel of data with a few latent variables, or factors, thus alleviating the curse of dimensionality. Despite its popularity and widespread use for various applications ranging from genomics to finance, this methodology has predominantly remained linear. This study estimates factors nonlinearly through the kernel method, which allows flexible nonlinearities while still avoiding the curse of dimensionality. We focus on factor-augmented forecasting of a single time series in a high-dimensional setting, known as diffusion index forecasting in macroeconomics literature. Our main contribution is twofold. First, we show that the proposed estimator is consistent and it nests linear PCA estimator as well as some nonlinear estimators introduced in the literature as specific examples. Second, our empirical application to a classical macroeconomic dataset demonstrates that this approach can offer substantial advantages over mainstream methods.
	}

% PCA & Factor Models, Forecasting & Prediction Methods, Neural Nets & Related Topics
\end{abstract}
\vfill

%%%%%%%%%%%%%%%%%%%%%%%%%%%%%%%%%%%%%%%%%%%%%%%%%%%%%%%%%%%%%%%%%%%%%%%%%%%%%%%%%%%%%%%
% INTRODUCTION %%%%%%%%%%%%%%%%%%%%%%%%%%%%%%%%%%%%%%%%%%%%%%%%%%%%%%%%%%%%%%%%%%%%%%%%
%%%%%%%%%%%%%%%%%%%%%%%%%%%%%%%%%%%%%%%%%%%%%%%%%%%%%%%%%%%%%%%%%%%%%%%%%%%%%%%%%%%%%%%
\newpage \spacingset{1.45}
\section{Introduction}
\label{sec:intro}

Over the past century, factor models have become an integral part of multivariate analysis and high-dimensional statistics, and have had a substantial effect on a number of different fields, including psychology (\cite{Thompson1938}), biology (\cite{Hirzel2002}) and economics (\cite{Chamberlain}). In economics and finance, applications range from portfolio optimization (\cite{Fama1992}) and covariance estimation (\cite{Fan2013}) to  forecasting (\cite{StockWatson2002a}). The application that we consider is macroeconomic forecasting, where various forms of factor analysis have become state-of-the-art techniques for prediction.

The general idea behind factor analysis consists of determining a few latent variables, or factors, that drive the dependence of the entire outcomes. Factors are designed to capture the common dynamics in a large panel of data.  This feature is crucial in the context of increasing availability of macroeconomic time series coupled with the inability of standard econometric methods to handle many variables. While classical econometric tools break down in such data-rich or ``big data" environments, factor models help to compress a large amount of the available information into a few factors, turning the curse of dimensionality into a blessing. 

Factor analysis possesses several attractive properties that justify a large amount of literature in its support. First, it effectively handles large dimensions thereby enhancing forecast accuracy in such regimes. This was demonstrated in \cite{StockWatson2002a} and \cite{StockWatson2002b} who used so-called diffusion indexes, or factors, in forecasting models when dealing with a large number of predictors. More recently, \cite{KimSwanson2018} find that factor augmented models nearly always outperform a wide range of big data and machine learning models in terms of predictive power. Second, due to its conceptual simplicity, this methodology found its use beyond academic research. For example, the Federal Reserve Bank of Chicago constructs the Chicago Fed National Activity Index (CFNAI) simply as the first principal component of a large number of time series. 
Third, factor analysis aligns naturally with the dynamic equilibrium theories as well as the stylized fact of \cite{Sargent77} of a small number of variables explaining most of the fluctuations in macroeconomic time series. And finally, factor estimates can also be used to provide efficient instruments for augmenting vector autoregressions (VARs) (\cite{Bernanke05}) to assist in tracing structural shocks.

Factors, however, are not observable and need to be estimated. Two classical estimation strategies rely on either intertemporal or contemporaneous smoothing. The former casts the model into a state-space representation and estimates it by the maximum likelihood via the Kalman filter. The disadvantages of this approach are that it requires parametric assumptions and that it quickly becomes computationally infeasible as the number of predictor series grows \footnote{Interestingly, there has been some evidence to the contrary, see \cite{Doz2012}}. Contemporaneous smoothing is a more predominant and computationally simpler way based on principal component analysis (PCA) (\cite{Pearson}), nonparametric least-squares approach for estimating factors. 	

Forecasts are obtained via a two-step procedure. First, factors estimates are derived from the set of available time series by one of the two methods described above. Once the factors are estimated, run a linear autoregression of the variable of interest onto factor estimates and observed covariates (e.g. lagged values of the dependent variable). 

We analyze a factor model that is high-dimensional, static and approximate. High-dimensional framework (\cite{BaiNg2002}), as opposed to classical framework (\cite{Anderson1984}), allows both time and cross-section dimensions to grow. Static models do not explicitly model time-dependence of factors contrary to more general dynamic counterparts (\cite{Forni2000a}). Approximate factor structure (\cite{Chamberlain}) is more flexible compared with a strict version (\cite{Ross1976}) as it imposes milder assumptions on the idiosyncratic component.

Factor analysis is closely related to PCA, although the two are not the same (\cite{Jolliffe1986}). It is, however, well documented that the two are asymptotically equivalent under suitable conditions (see the pervasiveness assumption in \cite{Fan2013} for a recent treatment).
There are several results on consistency of PCA estimators of factors (\cite{Connor1986}, \cite{StockWatson2002a}, \cite{Bai2006} among others) for various forms of factor models. One of the most relevant of the results is established in \cite{BaiNg2002} who derive convergence rates of such estimators for an approximate static factor model of large dimensions.

Despite its widespread use, factor modeling, and diffusion index forecasting methodology in particular, is still fundamentally limited to linear framework. Over the past two decades, leading researchers noted multiple times (e.g. see \cite{StockWatson2002b}, \cite{Bai2008}, \cite{StockWatson2012}, \cite{Hansen2015}) that further forecast improvements ``will need to come from models with nonlinearities and/or time variation" and that ``nonlinear factor-augmented regression should be considered" for forecasting. There have been attempts to incorporate time dependence (see, for example, \cite{DelNegro2008}, \cite{Mikkelsen2015} and \cite{Stevanovic2016} among others) as well as some work documenting the superiority of nonlinear models in time series context (\cite{Terasvirta1994}, \cite{Giovannetti}, \cite{KimSwanson2014}). However, the literature on addressing nonlinearity within the factor modeling or diffusion index methodology framework, which arguably remains to be the state-of-the-art technique for macroeconomic prediction (\cite{Stevanovich2019}), is scarce. 

One of the first such attempts is \cite{Amemiya2001} who assume nonlinear errors-in-variables parametrization of a factor model, which is not designed for forecasting applications. The most prominent work with focus on a prediction exercise is \cite{Bai2008}. They either augment the set of predictor time series with their squares and apply standard PC to the augmented set, or use squares of principal components obtained from the original (non-augmented) set. Another closely related work is \cite{Exterkate} who substitute the linear second step with kernel ridge regression and discover that this leads to more accurate forecasts of the key economic indicators.

This study adds to the scarce literature on nonlinear factor models. Specifically, the factors are allowed to capture nontrivial functions of predictors. To circumvent the computational difficulties associated with such novelties, we use the kernel trick, or kernel method (\cite{Hofmann2008}), which is discussed in the next section within the diffusion index methodology context.

The rest of the paper is organized as follows. Subsections 2.1 and 2.2 of Section 2 review the diffusion index methodology and the kernel trick; subsections 2.3 and 2.4 derive the estimators and provide the theoretical guarantees. Section 3 outlines the forecasting models, describes the data and provides the empirical results. Section 4 concludes and discusses possible extensions.  All proofs are given in the Appendix.

\textbf{Notation.} For a vector $\mathrm{v}\in\mathbb{R}^{d}$, we write its $i$-th element as $v_i$. The corresponding $\ell_p$ norm is $\norm{\mathrm{v}}_p = \big(\sum_{i=1}^d |v_i|^p\big)^{1/p}$, which is a norm for $1\le p \le \infty$. An inner product between two vectors of the same dimension is $\inner{\mathrm{v_i},\mathrm{v_j}} = \mathrm{v_i}'\mathrm{v_j}$. For a matrix $A\in\mathbb{R}^{m \times d}$, we write its $(i,j)$-th entry as $\{A\}_{ij} = a_{ij}$ and denote its $i$-th row (transposed) and $j$-th column as column vectors $A_{i\cdot}$ and $A_{\cdot j}$ respectively. Its singular values are $\sigma_1(A) \ge \sigma_2(A) \ge \ldots \ge \sigma_q(A)$, where $q=\min(m,d)$. The spectral norm is $\norm{A}_2 = \underset{\mathrm{v}\ne0}{\max} \frac{\norm{A\mathrm{v}}_2}{\norm{\mathrm{v}}_2}  = \sigma_1(A)$. The $\ell_1$ norm is $\norm{A}_1 = \underset{1\le j\le d}{\max} \sum_{i=1}^m |u_{ij}|$ and $\ell_\infty$ norm is $\norm{A}_\infty = \underset{1\le i\le m}{\max} \sum_{j=1}^d |u_{ij}|$. The Frobenius norm is $\norm{A}_F = \sqrt{\langle A,A \rangle} = \sqrt{tr(A'A)} = \sqrt{\sum_{i=1}^q \sigma_{i}^2(A)}$. For a symmetric matrix $W\in\mathbb{R}^{d\times d}$ with eigenvalues $\lambda_1(W) \ge \lambda_2(W) \ge  \ldots \ge \lambda_d(W)$, define $eig_r(W)\in\mathbb{R}^{d\times r}$ to be a matrix stacking $r\le d$ normalized eigenvectors in the order corresponding to $\lambda_1(W),\ldots, \lambda_r(W)$. Finally, for a sequence of random variables $\{X_n\}_{n=1}^\infty$ and a sequence of real nonnegative numbers $\{a_n\}_{n=1}^\infty$, denote $X_n = O_\mathbb{P}(a_n)$ if $\forall \epsilon >0, \exists M,N>0$ such that  $\forall n>N, \; \mathbb{P}(\abs{X_n/a_n}\ge M) < \epsilon$; and 
denote $X_n = o_\mathbb{P}(a_n)$ if $\forall \epsilon >0, \; \underset{n\to\infty}{\lim} \mathbb{P}(\abs{X_n/a_n}\ge \epsilon) = 0$. Finally, let $\mathbf{1}_{1/T}$ be a $T\times T$ matrix of ones divided by $T$.

%%%%%%%%%%%%%%%%%%%%%%%%%%%%%%%%%%%%%%%%%%%%%%%%%%%%%%%%%%%%%%%%%%%%%%%%%%%%%%%%%%%%%%%
% METHODOLOGY %%%%%%%%%%%%%%%%%%%%%%%%%%%%%%%%%%%%%%%%%%%%%%%%%%%%%%%%%%%%%%%%%%%%%%%%%
%%%%%%%%%%%%%%%%%%%%%%%%%%%%%%%%%%%%%%%%%%%%%%%%%%%%%%%%%%%%%%%%%%%%%%%%%%%%%%%%%%%%%%%

\section{Methodology}
\subsection{Diffusion Index Models}
Our goal is to accurately forecast a scalar variable $Y_t$, given a $T\times N$ data matrix $X$ with $t$th row $X_t'$, or $X_{t\cdot}'$. Both the number of observations $T$ and the number of series $N$ are typically large.

Consider the following baseline model, known as a Diffusion Index (DI) model:
\begin{equation}\label{eq:eq1}
\underset{1\times 1}{Y_{t+h}} = \underset{1\times r}{\beta_F'} \underset{r\times 1}{F_t} + \underset{1\times p}{\beta_W'} \underset{p\times 1}{W_t} + \underset{1\times 1}{\epsilon_{t+h}},
\end{equation}
\begin{equation}\label{eq:eq2}
\underset{N\times 1}{X_t} = \underset{N\times r}{\Lambda} \underset{r \times 1}{F_t} + \underset{N\times 1}{e_t}.
\end{equation}
Equation \ref{eq:eq1} is a linear forecasting model, where $Y_{t+h}$ is the value of the target variable $h$ periods in the future, $F_t$ is the vector of $r$ factors at time $t$, $W_t$ is a vector of $p$ observed covariates (e.g. an intercept and lags of $Y_{t+h}$), $\epsilon_{t+h}$ is a disturbance term. Equation \ref{eq:eq2} specifies the factor model, where $X_t$ is vector of $N$ candidate predictor series, $\Lambda$ is a loading matrix for $r$ common driving forces in $F_t$, $e_t$ is an idiosyncratic disturbance; and $t=1,\ldots,T$. The latter equation can be rewritten in matrix form
\begin{equation}\label{eq:eq3}
\underset{T\times N}{X} = \underset{T \times r}{F} \underset{r\times N}{\Lambda'} + \underset{T\times N}{e},
\end{equation}
where $X = [X_1,\ldots,X_T]'$ and $F = [F_1,\ldots,F_T]'$. Throughout the paper it is assumed that all series are weakly stationary, while variables in $X$ have been standardized, meaning that each variable is separately demeaned and set to have unit $\ell_2$-norm.

If the above set of equations is augmented with transition equations for $F_t$, we obtain a dynamic factor model which is estimated by the Kalman filter as discussed above. Let us instead focus on a nonparametric estimation approach as suggested in \cite{StockWatson2002a}. The goal at first stage is to solve 
\begin{equation}\label{eq:eq4}
	\begin{gathered}
	\underset{F,\Lambda}{\argmin} \; \norm{X - F\Lambda'}^2_F \\
	N^{-1}\Lambda'\Lambda = I_r, \quad F'F \text{ diagonal},
	\end{gathered}
\end{equation}

where the restrictions are in place for identifying the unique solution (up to a column sign change). It is well known that the estimator of factor loadings $\widehat{\Lambda}$ is given by the $r$ eigenvectors associated with largest eigenvalues of $X'X$, while $\widehat{F} = X\widehat{\Lambda}$. This estimator $\widehat{F}$ is equivalent to principal component (PC) scores derived from the matrix $X$. Once we have an estimate of $F$, the second stage involves least squares estimation of equation \ref{eq:eq1} with $F$ substituted with its estimate. 

It is clear that the standard PC estimator reduces the dimensionality of $X$ linearly: $F_t$ represents the projection of $X_t$ onto $r$ eigenvector directions exhibiting the most variation. However, if there is a nonlinearity in $X$, that is if the true lower dimensional representation is a nonlinear submanifold in the original space, such linear projections will be inaccurate. There are several ways to take into account a possible nonlinearity. For example, \cite{Bai2008} propose a squared principal components (SPC) procedure, which applies the standard PCA algorithm to the matrix $X$ augmented by its square, that is $[X,X^2]$. Although this procedure supposedly leads to additional forecasting gains, it is limited by the second-order features of the data. 

Other nonlinear dimension reduction techniques include Laplacian eigenmaps (\cite{Belkin2001}), Locally-Linear Embeddings (\cite{Roweis2000}), Isomaps (\cite{Tenenbaum2000}) and a number of others (\cite{Hastie1989}, \cite{Mark1991}, \cite{Hinton2008}). In this paper we use the approach that applies the kernel trick to the standard PCA, so\hyp{}called kernel PCA (kPCA) (\cite{Scholkopf1999}). This algorithm can be shown to contain a number of widely used dimensionality reduction methods, including the ones listed above (\cite{Hofmann2008}). While it permits modeling a set of nonlinearities rich enough for successful applications in nontrivial pattern recognition tasks such as face recognition (\cite{Kim2002}), the algorithm does not involve any iterative optimization.

\subsection{Kernel Method}
In this subsection we review the kernel trick methodology and illustrate its usefulness. 
The kernel method implicitly maps the original data nonlinearly into a high\hyp{}dimensional space, known as a feature space, $\varphi(\cdot): \mathcal{X} \rightarrow \mathcal{F}$. This space can in fact be infinite\hyp{}dimensional which would seemingly prohibit any calculations. However, the trick is precisely in avoiding such calculations.  The focus is instead on similarities between any two transformed data points $\varphi(\mathrm{x_i})$ and $\varphi(\mathrm{x_j})$ in the feature space\footnote{Formally, the feature space is thus a Hilbert space, that is a vector space with a dot product defined on it.} as measured by $\varphi(\mathrm{x_i})'\varphi(\mathrm{x_j})$, calculating which would at first sight require the knowledge of the functional form of $\varphi(\cdot)$. The solution is to use a kernel function $k(\cdot,\cdot): \mathcal{X} \times \mathcal{X} \rightarrow \mathbb{R}$, which would output the inner product in the feature space without ever requiring the explicit functional form of $\varphi(\cdot)$. Moreover, a valid kernel function guarantees the existence of a feature mapping $\varphi(\cdot)$ although its analytic form may be unknown. The only requirement for this is positive-definiteness of the kernel function (Mercer's condition, see \ref{app:Mercer}), specifically,  $\int \int f(\mathrm{x_i}) k(\mathrm{x_i}, \mathrm{x_j}) f(\mathrm{x_j})  d\mathrm{x_i}  d\mathrm{x_j} \ge 0$, for any square-integrable function $f(\cdot)$.

This kernel function forms a Gram matrix $K$, known as a Kernel matrix, elements of which are inner products between transformed training examples, that is $\{K\}_{ij} = k(\mathrm{x_i},\mathrm{x_j}) = \varphi(\mathrm{x_i})'\varphi(\mathrm{x_j})$. What makes the kernel trick useful is the fact that many models can be written exclusively in terms of dot products between data points. For example, the ridge regression coefficient estimator can be formulated as $X'(XX' + \lambda I_T)^{-1}\mathrm{y}$\footnote{This can be derived by solving the dual of the ridge least squares optimization problem, however a simpler approach would be to apply the matrix identity $(B'C^{-1}B + A^{-1})^{-1}B'C^{-1} = AB'(BAB'+C)^{-1}$ to the usual ridge estimator $(X'X + \lambda I_N)^{-1}X'\mathrm{y}$.}, and hence the prediction for a test example $\mathrm{x_*}$ is $\hat{y} = \mathrm{x_*}' X'(XX' + \lambda I_T)^{-1}\mathrm{y}$, where the dependence is exclusively on inner products between the covariates. This property allows us to apply the kernel method by substituting dot products between original variables with their nonlinear kernel evaluations, that is dot products between transformed variables. Hence, an alternative form is $ k_*'(K + \lambda I_T)^{-1}\mathrm{y}$, where $k_*$ with $\{k_*\}_i = \{\varphi(\mathrm{x}_*)'\varphi(X)'\}_{i}$ is a vector of similarities between the test example and training examples in the feature space. In terms of the time complexity, the algorithm needs to invert a $T\times T$ matrix instead of inverting an $N\times N$ matrix. 

The key advantage of the kernel method is that it effectively permits using a linear model in a high-dimensional nonlinear space, which amounts to applying a nonlinear technique in the original space.
As a toy example, consider a classification problem shown in Figure \ref{fig:reg1}, where the true function separating the two classes is a circle of radius .5 around the origin. The left panel depicts observations from two classes which are not linearly separable in the original two-dimensional space. Applying a simple polynomial kernel of degree $2$, $k(\mathrm{x_i}, \mathrm{x_j)} = (\mathrm{x_i}'\mathrm{x_j})^2$, implicitly corresponds to working in the feature space depicted on the right panel, since for $\varphi(\mathrm{x_i}) = (x_{i1}^2, \; \sqrt{2}x_{i1}x_{i2}, \; x_{i2}^2)'$ we have $\varphi(\mathrm{x_i})'\varphi(\mathrm{x_j}) = (\mathrm{x_i}'\mathrm{x_j})^2$. In this toy example, a linear classifier could perfectly separate the observations in the right panel of Figure \ref{fig:reg1}. 

\begin{figure}[!htb]\vspace{-20pt}
	\centering
	\includegraphics[width=1\linewidth]{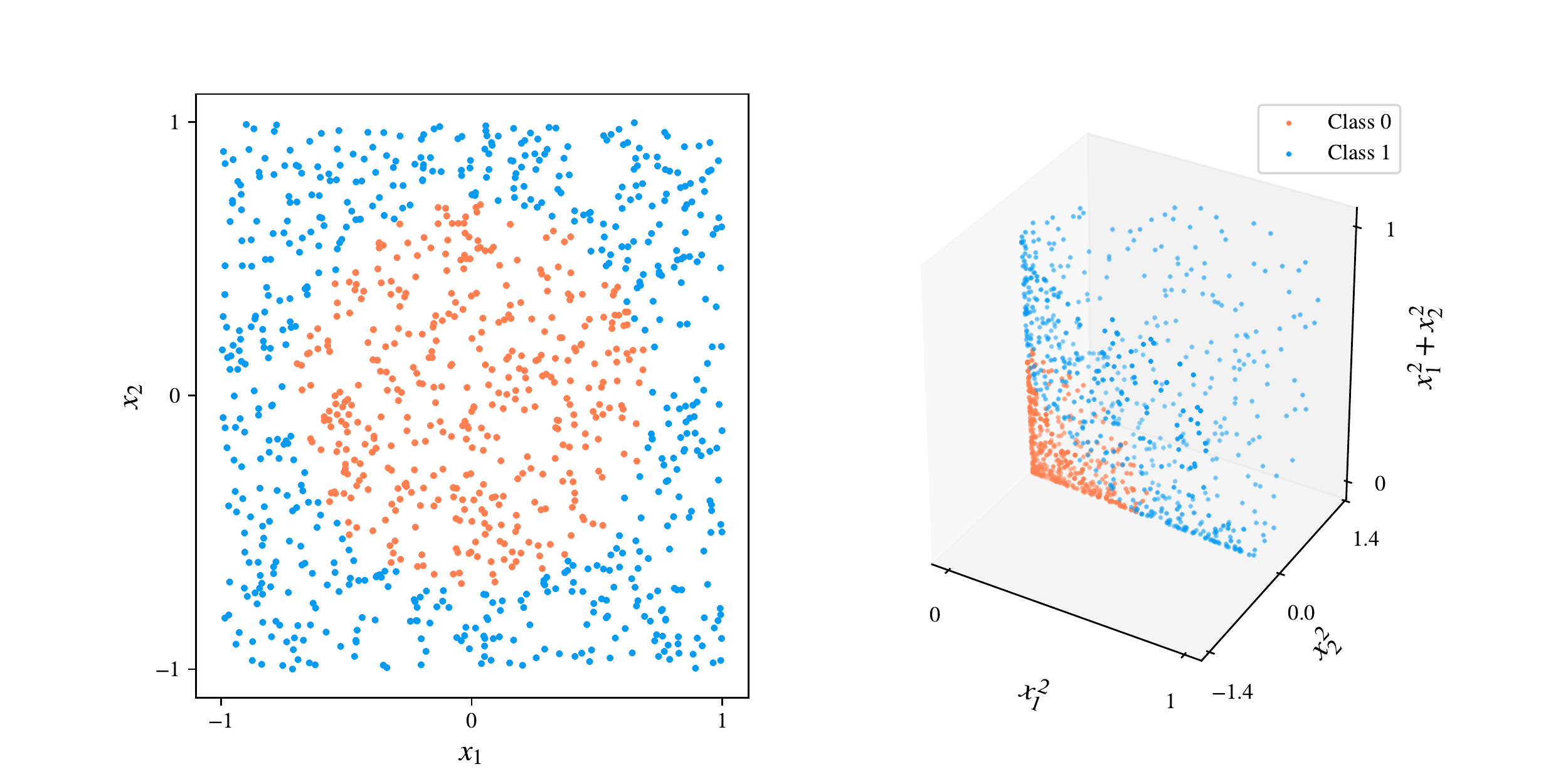}
	\caption{Kernel trick illustration for a toy classification example.}
	\medskip
	\small
	\textit{Left}: observations from two classes in the original space, not linearly separable. \textit{Right}: observations in the feature space, linearly separable. See the details in the text.
	\label{fig:reg1}
\end{figure}

While all valid kernel functions are guaranteed to have the corresponding feature space, in many cases it is implicit and infinite-dimensional, as, for instance, for the radial basis function (RBF) kernel $k(\mathrm{x_i}, \mathrm{x_j}) = e^{-\gamma\norm{\mathrm{x_i} - \mathrm{x_j}}_2^2}$ (see \ref{app:rbf}). Again, luckily, the knowledge of the feature mapping is not required.

%%%%%%%%%%%%%%%%%%%%%%%%%%%%%%%%%%%%%%%%%%%%%%%%%%%%%%%%%%%%%%%%%%%%%%%%%%%%%%%%%%%%%%%
% kPCA %%%%%%%%%%%%%%%%%%%%%%%%%%%%%%%%%%%%%%%%%%%%%%%%%%%%%%%%%%%%%%%%%%%%%%%%%%%%%%%%
%%%%%%%%%%%%%%%%%%%%%%%%%%%%%%%%%%%%%%%%%%%%%%%%%%%%%%%%%%%%%%%%%%%%%%%%%%%%%%%%%%%%%%%
\subsection{Nonlinear Modeling and kPCA}

Suppose there is a nonlinear function $\varphi(\cdot): \mathbb{R}^N \rightarrow \mathbb{R}^M$, where $M \gg N$ is very large (often infinitely large), mapping each observation to a high-dimensional feature space, $X_t \rightarrow \varphi(X_t)$. For now we consider $M$ to be finite for simplicity of exposition, so the original $T\times N$ data matrix $X$ can be represented as a $T\times M$ matrix $\Phi = [\varphi(X_1),\ldots,\varphi(X_T)]'$ in the transformed space, which may not be observable. Infinite-dimensional case induces several complications and is considered later. 

Both the original $X$ and its transformation $\Phi$ are assumed to be demeaned. The latter requirement is simple to incorporate in the kernel matrix despite the mapping being unobserved. Specifically, supposing the original (non-demeaned) transformation is $\widetilde{\Phi}$, the kernel associated with demeaned features is 
\begin{equation}\label{eq:eq5}
K = (I_T - \mathbf{1}_{1/T})\widetilde{\Phi} \widetilde{\Phi}'(I_T - \mathbf{1}_{1/T})' = \widetilde{K} - \mathbf{1}_{1/T}\widetilde{K} - \widetilde{K}\mathbf{1}_{1/T} + \mathbf{1}_{1/T}\widetilde{K}\mathbf{1}_{1/T},
\end{equation}
where $\widetilde{K} = \widetilde{\Phi} \widetilde{\Phi}'$ is based on the original $\widetilde{\Phi}$.

Our modeling of nonlinearity is through the feature mapping $\varphi(\cdot)$. This function replaces the original variables of interest in equation \eqref{eq:eq2} with their transformations,
\begin{equation}\label{eq:eq6}
\underset{M\times 1}{\varphi(X_t)} = \underset{M\times r}{\Lambda_\varphi} \underset{r \times 1}{F_{\varphi,t}} + \underset{M\times 1}{e_{\varphi,t}},
\end{equation}
where the subscript $\varphi$ indicates the association with the transformation. By stacking these into a  $T\times M$ matrix $\Phi$ we can rewrite the minimization problem \eqref{eq:eq4} as
\begin{equation}\label{eq:eq7}
	\begin{gathered}
	\underset{F_\varphi,\Lambda_\varphi}{\argmin} \; \norm{\Phi - F_\varphi\Lambda_\varphi'}^2_F \\
	N^{-1}\Lambda_\varphi'\Lambda_\varphi = I_r, \quad F_\varphi'F_\varphi  \text{ diagonal}.
	\end{gathered}
\end{equation}
Note that solving this directly through the eigendecomposition of $\Phi'\Phi$ is generally infeasible, since $\Phi'\Phi$ is $M\times M$ dimensional. Even if the dimension $M$ was not prohibitive, the map $\varphi(\cdot)$ is unknown for interesting problems rendering any computation dependent on $\Phi$ or $\Phi'\Phi$ alone impossible. Fortunately, it is possible to reformulate this problem in terms of the $T\times T$ Gram matrix $K = \Phi\Phi'$.

While we are assuming $M$ to be prohibitively large but finite, the following decomposition generalizes to infinite dimensions. 
Starting from the ``infeasible" eigendecomposition of the unknown covariance matrix of $\Phi$
\begin{equation}\label{eq:eq8}
\frac{\Phi'\Phi}{T}V_\varphi^{[i]} = \lambda_i^c V_\varphi^{[i]}, \qquad i=1,\ldots,M,
\end{equation}
where the eigenvalues $\lambda_i^c = \lambda_i(\frac{\Phi'\Phi}{T})$ satisfy $\lambda_1^c \ge \lambda_2^c \ge \ldots \ge \lambda_T^c$ and  $\lambda_j^c=0$ for $j > T$ (assuming $M\ge T$) and $V^{[i]}$ is an $M$-dimensional eigenvector associated with the $i$th eigenvalue $\lambda_i^c$.

The key is to observe that each $V^{[i]}_\varphi$ can be expressed as a linear combination of features
\begin{equation}\label{eq:eq9}
V^{[i]}_\varphi = \frac{\Phi'\Phi}{\lambda_i^c T} V^{[i]}_\varphi \equiv \Phi'A^{[i]}, \qquad i=1,\ldots,M,
\end{equation}
where $A^{[i]} = \frac{\Phi V^{[i]}_\varphi}{\lambda_i^c T} = \left[ \alpha_1^{[i]}, \; \ldots \;,  \alpha_T^{[i]} \right]'$ is a vector of weights which is determined next. Plugging this back into \eqref{eq:eq8} yields
\begin{equation} \label{eq:eq10}
\lambda_i^c \Phi'A^{[i]} = \frac{\Phi'\Phi}{T} \Phi'A^{[i]}, \qquad i = 1,\ldots,M.
\end{equation}
Finally, premultiplying equation \eqref{eq:eq10} to the left by $\Phi$ and removing $K=\Phi\Phi'$ from both sides we obtain
\begin{equation}
\frac{K}{T} A^{[i]} = \lambda_i^c A^{[i]}, \qquad i = 1,\ldots,M,
\end{equation}
hence the $i$th vector of weights $A^{[i]}$ corresponds to an eigenvector of a finite-dimensional Gram matrix $K$ associated with the $i$th largest eigenvalue $\lambda_i(\frac{K}{T}) = \lambda_i^c$  with  $\lambda_j(\frac{K}{T})=0$ for $j > T$. 

Notice that while solving the eigenvalue problem of $\frac{K}{T}$ allows to compute $A^{[i]}$, we are still unable to obtain the vector $V^{[i]} = \Phi'A^{[i]}$ since $\Phi$ may not be known. However, the main object of interest is recoverable: to calculate principal component projections, we  project the data onto (unknown) eigenspace,
\begin{equation}\label{eq:kernelfactor}
\underset{T\times 1}{\widehat{F}^{[i]}_\varphi} = \Phi V^{[i]} = \Phi \Phi'A^{[i]}  = KA^{[i]}, \qquad i = 1,\ldots,M.
\end{equation}
Stacking estimated factors corresponding to the first $r$ eigenvalues, define a $T\times r$ matrix $\widehat{F}_\varphi = \left[\widehat{F}^{[1]}_\varphi, \; \ldots\;, \widehat{F}^{[r]}_\varphi \right]$, where the subindex $r$ is dropped to simplify the notation. We refer to the factors constructed this way as \textit{kernel factors}. 

A similar alternative solution that only involves the Gram matrix has been known in econometrics since at least \cite{Connor1993}. In particular, for a given $r$ the optimization problem in \eqref{eq:eq7} with the identification constraints $T^{-1}F_\varphi'F_\varphi = I_r$ and diagonal $\Lambda_\varphi'\Lambda_\varphi$, has the solution $\widetilde{F}_\varphi = \sqrt{T} eig_r(\Phi \Phi') = \sqrt{T} A_r$, where $A_r = \left[\widehat{A}^{[1]}, \; \ldots\;, \widehat{A}^{[r]} \right]$. Hence, the kPCA estimator is equivalent to the latter premultiplied by  $\frac{\Phi \Phi'}{\sqrt{T}} = \frac{K}{\sqrt{T}}$. Note, however, that both estimators yield the same predictions when passed to the main forecasting equation \eqref{eq:eq1} as they have identical column spaces. This idea is summarized in the following proposition. 

\begin{proposition}
\label{prop1}
Estimators $\widehat{F}_\varphi = \Phi\Phi'eig_r(\Phi \Phi')$ and $\widetilde{F}_\varphi = \sqrt{T} eig_r(\Phi \Phi')$ produce the same projection matrix.
\end{proposition}

Importantly, we also establish that certain commonly used kernels allow the kernel factor estimator to incorporate the usual PC estimator. The following proposition demonstrates that RBF and sigmoid kernels allow to nest (a constant multiple of) the PC estimator for limiting values of the hyperparameter.

\begin{proposition}
	\label{prop2}
	For a column-centered matrix $X \in \mathbb{R}^{T\times N}$, let $\widehat{F}_\varphi = K eig_r(K)$ be the kernel factor estimator and $\widehat{F} = X eig_r(X'X)$ be the usual linear PCA factor estimator. Then \vspace{-5pt}
	\[\exists s= \pm 1, \text{ such that } \underset{\gamma \to 0}{\lim} \; c\gamma^{-1} \widehat{F}_\varphi L^{-1/2} = s\widehat{F}, \quad \forall r=1,\ldots,\min{\{T,N\}}, \vspace{-5pt}\] 
	 where $K = \widetilde{K} - \mathbf{1}_{1/T}\widetilde{K} - \widetilde{K}\mathbf{1}_{1/T} + \mathbf{1}_{1/T}\widetilde{K}\mathbf{1}_{1/T}$, $L$ is a diagonal matrix of $r$ largest eigenvalues of $XX'$ sorted in nonincreasing order. Furthermore, \\
	 \textbf{\textrm{(a)}} for RBF kernel $\widetilde{k}_{ij} = e^{-\gamma\norm{X_{i\cdot} - X_{j\cdot}}_2^2}$ we have $c=2^{-1}$,\\
	 \textbf{\textrm{(b)}} for sigmoid kernel $\widetilde{k}_{ij} = \tanh(c_0+\gamma X_{i\cdot}'X_{j\cdot})$ we have $c=(1-\tanh^2(c_0))^{-1}$, where $c_0$ is an arbitrary (hyperparameter) constant.
\end{proposition}
\begin{proof}
	See appendix \ref{app:prop2}.
\end{proof}

Adjustment by $\gamma^{-1}$ is necessary as the entries of $\widehat{F}_\varphi = K eig_r(K)$ approach $0$ as $\gamma \to 0$. This is due to the nature of $K = \widetilde{K} - \mathbf{1}_{1/T}\widetilde{K} - \widetilde{K}\mathbf{1}_{1/T} + \mathbf{1}_{1/T}\widetilde{K}\mathbf{1}_{1/T}$, where $k_{ij}$ converge to $0$ for both kernels for all $i,j=1,\ldots,T$. Luckily $\gamma^{-1}K$ has a nice non-degenerate limiting value (shown in the proof) that ensures the limit in proposition \ref{app:prop2} holds. 

Proposition \ref{prop2} states that the (properly scaled) kernel factor matrix converges pointwise to its PCA analog (up to a sign flip) as the value of the hyperparameter $\gamma$ nears zero. That is, in the limit the two factor estimators are constant multiples of one another and hence produce the same forecasts. This ability to mimic the linear estimator is important for certain applications. For example, in macroeconomic forecasting it is notoriously difficult to beat linear models in a short horizon prediction exercise. 

Figure \ref{fig:reg2} illustrates the implicit procedure for obtaining $r$ kernel factors. The selected kernel function induces nonlinearity $\varphi(\cdot)$ on each element of the input layer, $N$-dimensional observations $X_{1\cdot}, \ldots, X_{T\cdot}$. Next, pairwise similarities between high\hyp{}dimensional vectors $\varphi(X_{1\cdot}), \ldots, \varphi(X_{T\cdot})$ are computed, with $k(X_{i\cdot},\cdot) = \Big[ \varphi(X_{i\cdot})'\varphi(X_{1\cdot}),\allowbreak \; \ldots \;, \allowbreak \varphi(X_{i\cdot})'\varphi(X_{T\cdot}) \Big]'$. Finally, each factor is obtained as a linear combination of these inner products with the weights given as a solution to the eigenvalue problem discussed above. Of course, the kernel PCA algorithm does not explicitly nonlinearize the inputs as the kernel trick allows us to immediately calculate similarities and avoid the expensive high-dimensional computation. However, as mentioned earlier, the existence of such implicit nonlinearities is guaranteed by Mercer's theorem (\ref{app:Mercer}). As opposed to standard feedforward neural networks, there is no iterative training involved and no peril of being trapped in local optima. On the other hand, it has been noted that certain kernels, including RBF and sigmoid, allow extracting features of the same type as the ones extracted by neural networks (\cite{Scholkopf1999}). The two necessary steps involve evaluation of similarities in the kernel matrix and solving its eigenvalue problem; the complexity is thus dependent only on the sample size. 

\begin{figure}[!htb]
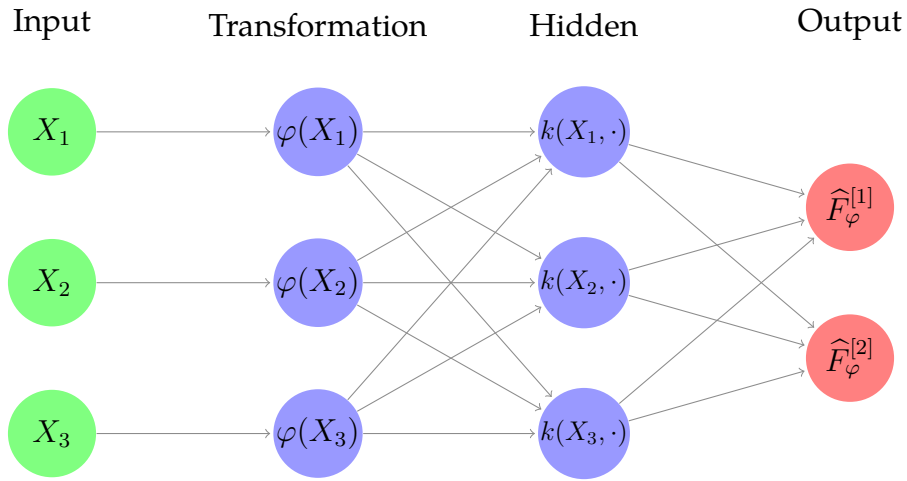

	\centering
	\begin{neuralnetwork}[height=4, layertitleheight=1em, layerspacing=35mm, nodespacing=20mm, nodesize=33pt]
		\newcommand{\nodetextclear}[2]{}
		\newcommand{\nodetextx}[2]{$X_#2$}
		\newcommand{\nodetextphi}[2]{$\varphi(X_{#2})$}
		\newcommand{\nodetextk}[2]{$\Scale[.85]{k(X_{#2},\cdot)}$}
		\newcommand{\nodetextf}[2]{$\widehat{F}_\varphi^{[#2]}$}
		\inputlayer[count=3, bias=false, title=Input, text=\nodetextx]
		\hiddenlayer[count=3, bias=false, title=Transformation, text=\nodetextphi] 
		\link[from layer=0, to layer=1, from node=1, to node=1]
		\link[from layer=0, to layer=1, from node=2, to node=2]
		\link[from layer=0, to layer=1, from node=3, to node=3]
		\hiddenlayer[count=3, bias=false, title = Hidden, text=\nodetextk] \linklayers 
		\outputlayer[count=2, title=Output, text=\nodetextf] \linklayers
	\end{neuralnetwork}
	\caption{Neural network interpretation of kernel PCA, $T=3$, $r=2$.}
	\medskip
	\small
	Each observation is nonlinearly transformed and the inner products are computed. The output units are kernel factors with $\widehat{F}_\varphi^{[j]} = \sum_{i=1}^T \alpha_i^{[j]}k(X_i,\cdot)$ which linearly combine these dot products, with the weight estimate calculated as the eigenvector of the kernel matrix $K$.
	\label{fig:reg2}
\end{figure}

%%%%%%%%%%%%%%%%%%%%%%%%%%%%%%%%%%%%%%%%%%%%%%%%%%%%%%%%%%%%%%%%%%%%%%%%%%%%%%%%%%%%%%%
% Theory %%%%%%%%%%%%%%%%%%%%%%%%%%%%%%%%%%%%%%%%%%%%%%%%%%%%%%%%%%%%%%%%%%%%%%%%%%%%%%
%%%%%%%%%%%%%%%%%%%%%%%%%%%%%%%%%%%%%%%%%%%%%%%%%%%%%%%%%%%%%%%%%%%%%%%%%%%%%%%%%%%%%%%
\subsection{Theory}
Depending on the choice of the kernel, the induced feature space could be either finite or infinite. A polynomial kernel considered earlier generates a finite-dimensional feature space, consisting of a set of polynomial functions over the inputs. In this simple case, the eigenspace associated with the kernel factor estimator in equation \eqref{eq:kernelfactor} can generally be consistently estimated within the framework of \cite{Bai2003}. Particularly, proposition \ref{prop1} and the following theorem in \cite{Bai2003} immediately imply $\sqrt{M}$-consistency of $\widetilde{F}_\varphi$.

For the model associated with equation \eqref{eq:eq6}:

\textbf{Assumption A}: There exists a constant $c_1 < \infty$ independent of $M$ and $T$, such that
\begin{itemize}\label{ass:Ass1}\setlength\itemsep{-.5em}
	\item[] (a) $\E\norm{F_{\varphi,t}}_F^4 \le c_1$ and $T^{-1} F_{\varphi}'F_{\varphi} \overset{p}{\to} \Sigma_F > 0$, where $\Sigma_F$ is a non-random positive definite matrix;
	\item[] (b) $\E\norm{\Lambda_{\varphi,i\cdot}}_F \le c_1$ and $N^{-1} \Lambda_\varphi'\Lambda_\varphi \overset{p}{\to} \Sigma_\Lambda > 0$, where $\Sigma_\Lambda$ is a non-random positive definite matrix.
	\item[] (c1) $\E(e_{\varphi,it}) = 0$, $\E\abs{e_{\varphi,it}}^8 \le c_1$;
	\item[] (c2) $\E(e_{\varphi,s}'e_{\varphi,t}/M) = \gamma_M(s,t), \; \abs{\gamma_M(s,s)}\le c_1 \; \forall s$, $T^{-1}\sum_{s=1}^T\sum_{t=1}^T\abs{\gamma_M(s,t)} \le c_1$, $\sum_{s=1}^T \gamma_M(s,t)^2 \le M \; \forall t,T$;
	\item[] (c3) $\E(e_{\varphi,it}e_{\varphi,jt}) = \tau_{ij,t}$ with $\abs{\tau_{ij,t}} \le \abs{\tau_{ij}}$ for some $\tau_{ij}$ and $\forall t$; and $M^{-1}\sum_{i=1}^M \sum_{j=1}^M \allowbreak\abs{\tau_{ij}}\le c_1$;
	\item[] (c4) $\E(e_{\varphi,it}e_{\varphi,js}) = \tau_{ij,ts}$, $(MT)^{-1} \sum_{i=1}^M \sum_{j=1}^M \sum_{t=1}^T \sum_{s=1}^T \abs{\tau_{ij, ts}} \le c_1$;
	\item[] (c5) $\E \abs{M^{-1/2} \sum_{i=1}^M \left(e_{\varphi,is} e_{\varphi,it} - \E(e_{\varphi,is} e_{\varphi,it})\right)}^4 \le c_1 \; \forall t,s$;
	\item[] (d) $\E\left(\frac{1}{M} \sum_{i=1}^M \norm{\frac{1}{\sqrt{T}} \sum_{t=1}^T F_{\varphi, t}e_{\varphi,it}}_F^2 \right) \le c_1 $;
	\item[] (e) $T^{-1} F_{\varphi}'F_{\varphi} = I_r$ and $\Lambda_\varphi'\Lambda_\varphi$ is diagonal with distinct entries. 
\end{itemize}
Part (a) is standard in factor model literature. Part (b) ensures pervasiveness of factors in the sense that each factor has non-negligible effect on the variability of covariates. This is crucial for asymptotic identification of the common and idiosyncratic components. Parts (c$\cdot$) partially permit time-series and cross-section dependence in the idiosyncratic component, as well as heteroskedasticity in both dimensions. Possible correlation of $e_{\varphi,it}$ across $i$ sets up the model to have an approximate factor structure. Part (d) allows weak dependence between factors and idiosyncratic errors and (e) is for identification. 

The following theorem establishes consistency of kernel factors, for kernels inducing finite nonlinearities, in a large-dimensional framework.

\begin{theorem}[Theorem 1 \cite{BaiNg2002}, adapted] 
	\label{thm:thm1}
	Suppose the kernel function induces finite dimensional nonlinearity, i.e. for $ N<\infty,\; \varphi(\cdot): \mathbb{R}^N \to \mathbb{R}^M$, where $M:=M(N)$ is such that $ N\le M(N)<\infty$, and \hyperref[ass:Ass1]{Assumption A} holds. Then for any fixed $r\ge 1$, as $M,T \to \infty$
	\begin{equation*}
	\delta_{NT}^2\norm{\widetilde{F}_{\varphi,t} - H'F_{\varphi,t}}_F^2 = O_p(1), \quad \forall t=1,\ldots,T,
	\end{equation*}

	where $\delta_{NT} = \min\{\sqrt{M},\sqrt{T}\}$, $\underset{r\times r}{H} = \frac{\Lambda^0_\varphi\vspace{0pt}'\Lambda^0_\varphi}{M}\frac{F^0_\varphi\vspace{0pt}'\widetilde{F}_\varphi}{T}V_{MT}^{-1}$, $\underset{r \times r}{V_{MT}}$ is a diagonal matrix of $r$ largest eigenvalues of $\frac{\Phi \Phi'}{MT}$, $\widetilde{F}_{\varphi,t}$ and $F_{\varphi,t}$ are $t$-th rows in $\widetilde{F}_\varphi = \sqrt{T} eig_r(\Phi \Phi')$ and $F_\varphi$, respectively.
\end{theorem} 
\begin{proof}
	See appendix \ref{app:thm1}.
\end{proof}

Some comments are in order. First, the theorem states that the squared differences between the proposed factor estimator and (a rotation of) the true factor vanish as $M,T \to \infty$. While true factors themselves are not identifiable unless additional assumptions are imposed (see \cite{BaiNg13}), identification of the latent space spanned by factors is just as good as exact identification for forecasting purposes. 
Second, since for any given number of original variables $N$ the dimension of the transformed space $M$ is fixed and finite, the growth in $M$ is only possible through $N$ and thus the limit on $M$ implies one on $N$.
Third, this result does not imply uniform convergence in $t$.
Lastly, the results suggest the possibility of $\sqrt{T}$-consistent estimation of the forecasting equation with respect to its conditional mean.

Unfortunately, this result does not generalize to the most interesting kernels (e.g. RBF) inducing infinite-dimensional Hilbert spaces, rendering the traditional approach unsuitable for establishing theoretical properties. Hence we turn to a functional analytic framework which allows rigorous treatment of infinite-dimensional spaces and which has been the classical framework for analyzing statistical properties of functional PCA, and kernel PCA in particular, in machine learning literature. 

As will be shown later, it turns out that we can still show that the estimator concentrates around its population counterpart. We briefly present the necessary terms for understanding this result, without aiming to be exhaustive. A sufficiently detailed introduction to the analysis in Hilbert spaces can be found in \cite{Blanchard06}, while a classical reference for operator perturbation theory is \cite{Kato52}. 

One of the first major investigations of the statistical properties of kernel PCA can be found in \cite{Shawe-Taylor02}. The study provides concentration bounds on the sum of eigenvalues of the kernel matrix towards that of corresponding (infinite-dimensional) kernel operators. This permits to characterize the accuracy of kPCA in terms of the reconstruction error, that is the ability to preserve the information about a high-dimensional input in low dimensions. This is of significant interest for certain types of applications, such as pattern recognition. \cite{Blanchard06} further extend the results of the aforementioned study and improve the bounds on eigenvalues using tools of perturbation theory. 

Although the theoretical discussion and the mathematical approaches developed in this literature are extremely valuable, our interest is not in kPCA's ability to reconstruct a given observation. The kernel factor estimator only reduces the dimensionality and passes it to the next stage, without ever going through the reconstruction phase. As the dimensionality is reduced by projecting observations onto the eigenspace -- the space spanned by eigenfunctions of the true covariance operator with largest eigenvalues --  the interest is in convergence of empirical eigenfunctions towards the true counterparts. Importantly, the proximity of eigenvalues does not guarantee that underlying eigenspaces will also be close. 

We now briefly introduce the technical background for understanding our result. 
Let $\mathcal{H}$ be an inner product space, that is a linear vector space endowed with an inner product $\inner{\cdot,\cdot}_\mathcal{H}$, commonly denoted together as $(\mathcal{H}, \inner{\cdot,\cdot}_\mathcal{H})$. An inner product space $(\mathcal{H}, \inner{\cdot,\cdot}_\mathcal{H})$ is a Hilbert space if and only if the norm induced by the inner product $\norm{\cdot}_\mathcal{H} = \inner{\cdot. \cdot}_\mathcal{H}^{1/2}$ is complete\footnote{Specifically, the limits of all Cauchy sequences of functions must be in the Hilbert space.}. The results below rely on this norm, so we drop the subscript $\mathcal{H}$ to simplify the notation.

Let $X$ be a random variable taking values on a general probability space $\mathcal{X}$. 
A function $k: \mathcal{X} \times \mathcal{X} \to \mathbb{R}$ is a kernel if there exists a real-valued Hilbert space and a measurable feature mapping $\phi: \mathcal{X} \to \mathcal{H}$ such that $\forall x,x' \in \mathcal{X}, \; k(x,x') = \inner{\phi(x), \phi(x')}.$ A function $k: \mathcal{X} \times \mathcal{X} \to \mathbb{R}$ is a reproducing kernel of $\mathcal{H}$ and $\mathcal{H}$ is a reproducing kernel Hilbert space (RKHS), if $k$ satisfies (i) $\forall x \in \mathcal{X}, k(\cdot,x) \in \mathcal{X}$, and the reproducing property (ii) $\forall x \in \mathcal{X}, \forall f \in \mathcal{H}, \inner{f,k(\cdot,x)} = f(x)$. An important result from \cite{Aronszajn50} guarantees the existence of a unique RKHS for every positive definite $k$. 

Assume that $\E \phi(X) = 0$ and $\E\norm{\phi(X)}^2<\infty$. A unique covariance operator on $\phi(X)$, $\Sigma = \E \phi(X) \otimes \phi(X)$, satisfying $\inner{g,\Sigma h}_\mathcal{H} = \E \inner{h,\phi(X)}_\mathcal{H} \inner{g,\phi(X)}_\mathcal{H},\; \forall g,h \in \mathcal{H}$, always exists (Theorem 2.1 in \cite{Blanchard06}) and is a positive, self-adjoint trace-class operator. Denote  $\widehat{\Sigma} = \frac{1}{T} \sum_{i=1}^{T}\phi(X_i)\otimes \phi(X_i)$ to be its empirical counterpart. Finally, an orthogonal projector in $\mathcal{H}$ onto a closed subspace $V$ is an operator $\Pi_V$ such that $\Pi_V^2 = \Pi_V$ and $\Pi_V = \Pi_V^*$.

We now lay out two lemmas which lead to the result. Lemma \eqref{lem:lem1} bounds the difference between the true and empirical covariance operators.

\begin{lemma}[Difference between sample and true covariance operators]
	\label{lem:lem1} \text{}\\
	Assume random variables $X_1, \ldots, X_T \in \mathcal{X}$ are independent and $\sup_{x\in\mathcal{X}} k(x,x) \le \bar{k}$, then
	\begin{equation*}
	\mathbb{P} \left( \norm{\widehat{\Sigma} - \Sigma} \ge \left( 1+\sqrt{\frac{\epsilon}{2}} \right) \frac{2\bar{k}}{\sqrt{T}} \right) \le e^{-\epsilon}.
	\end{equation*}
\end{lemma}
\begin{proof}
	See appendix \ref{app:lem1}.
\end{proof}

Note that both sigmoid and RBF kernels are bounded and hence satisfy the requirement of the above lemma. Lemma \eqref{lem:lem2} is an operator perturbation theory result and is adapted from \cite{Koltchinskii2000} and \cite{Zwald}:
\begin{lemma}\label{lem:lem2}
	Let $A,B$ be two symmetric linear operators. Denote the distinct eigenvalues of $A$ as $\mu_1 > \ldots>\mu_k>0$ and let $\Pi_i$ be the orthogonal projector onto the i-th eigenspace. For a positive integer $p\le k$  define $\delta_p(A) := \min \{\abs{\mu_i - \mu_{j}}: 1\le i < j \le p+1 \}$. Assuming $\norm{B} < \delta_p(A)/4$, then 
	\[ \norm{\Pi_i(A) - \Pi_i(A+B)} \le \frac{4\norm{B}}{\delta_i(A)}. \]
\end{lemma}
\begin{proof}
	See appendix \ref{app:lem2}.
\end{proof}

Finally, the following theorem bounds the difference between empirical and true eigenvectors.

\begin{theorem}\label{thm:thm2}
	Denote the $i$-th eigenvectors of $\widehat{\Sigma}$ and $\Sigma$ as $\widehat{\psi}_i$ and $\psi_i$ respectively. Then, under the assumptions of Lemma \ref{lem:lem1} and \ref{lem:lem2}, as $T\to \infty$ we have
	\[ \norm{\widehat{\psi}_i-\psi_i} = o_p(1)\]
\end{theorem}
\begin{proof}
	See appendix \ref{app:thm2}.
\end{proof}

Some comments are in order. First, note that we require the eigenvalues to be distinct, a well-known restriction (similar to $sin(\theta)$ theorem of \cite{Davis70}), since it is impossible to identify eigenspaces with the same eigenvalues. Second, Theorem \ref{thm:thm2} suggests that the eigenspace estimated by kernel PCA will concentrate close to the true eigenspace. Our kernel factor estimator simply projects onto that eigenspace and hence the precision is expected to increase as $T\to\infty$. Third, this rate does not address the case when variables exhibit dependence, although the exercise in the next section is indicative of some form of concentration, which suggests that the theoretical assumptions might be too conservative. Lastly, it may be possible to obtain a sharper bound since the proof relies on crude inequalities (e.g. triangle inequality).

%%%%%%%%%%%%%%%%%%%%%%%%%%%%%%%%%%%%%%%%%%%%%%%%%%%%%%%%%%%%%%%%%%%%%%%%%%%%%%%%%%%%%%%
% Empirical Evaluation %%%%%%%%%%%%%%%%%%%%%%%%%%%%%%%%%%%%%%%%%%%%%%%%%%%%%%%%%%%%%%%%
%%%%%%%%%%%%%%%%%%%%%%%%%%%%%%%%%%%%%%%%%%%%%%%%%%%%%%%%%%%%%%%%%%%%%%%%%%%%%%%%%%%%%%%

\section{Empirical Evaluation}

\subsection{Forecasting Models}
This subsection discusses specific forms of equations \eqref{eq:eq1} and \eqref{eq:eq2} that are used for forecasting. Autoregressive Diffusion Index (ARDI) model is specified as
\begin{equation}
	Y_{t+h} = \beta_0^h + \sum_{p=1}^{P_t^h} \beta_{Y,p}^h Y_{t-p+1} + \sum_{m=1}^{M_t^h} \underset{1\times K_t^h}{\beta_{F,m}^{h}}\hspace{-2pt}' \underset{\hspace{-5pt}K_t^h\times 1}{F_{t-m+1}} + \epsilon_{t+h},
\end{equation}
where superscript $h$ indicates dependence on the time horizon. Note, $P_t^h, M_t^h, K_t^h$ are the number of lags of the target variable, the number of lags of factors, the number of factors respectively. These three parameters are estimated simultaneously for each time period and time horizon using BIC. Since the true factors are unknown, we instead plug in the estimates from the factor equation, which is discussed next.

Three factor equation specifications are considered,
\begin{equation}\label{eq:eq15}
	\underset{N\times 1}{X_t} = \Lambda F_t + e_t,
\end{equation} \vspace{-30pt}
\begin{equation}\label{eq:eq16}
	\underset{2N\times 1}{X_{*,t}} = \Lambda_* F_{*,t} + e_{*,t},
\end{equation}
\begin{equation}\label{eq:eq17}
	\underset{M\times 1}{\varphi(X_t)} = \Lambda_\varphi F_{\varphi,t} + e_{\varphi,t}.
\end{equation}
Factors in equation \eqref{eq:eq15} are estimated by PCA. Equation \eqref{eq:eq16} is similar, replacing the left-hand side with an augmented vector $X_{*,t}=[X_t,X_t^2]$. This procedure was dubbed as squared principal components (SPC) in \cite{Bai2008}. Finally, the last equation applies nonlinearity induced by the selected kernel and is estimated by kPCA.

Forecasts using PCA and SPC are produced in three steps. First, we extract three factors from the set transformed and standardized predictors using one of the two methods. Second, three parameters are determined according to BIC for each out-of-sample forecasting period and each prediction horizon: the number of lags of the target variable $P^h_t$, the number of lags of factors $M^h_t$, the number of factors $K^h_t$. Third, the forecasting equation is estimated by least squares and forecasts are produced. The procedure for predicting with kPCA is similar, except there is an additional step where the value of the hyperparameter is specified, and the estimation is instead made in accordance with Algorithm \ref{alg:alg1}.

\begin{algorithm}[htb]
	\setstretch{1}
	\SetAlgoLined
	\KwIn{Observations $X_1,\ldots,X_T \in \mathbb{R}^N$, kernel function $k(\cdot,\cdot)$, dimension $r$.}
	\For{$i=1,\ldots,T$}{
		\For{$j=1,\ldots,T$}{
			$K_{ij} = k(X_i,X_j)$ \tcc*[r]{Compute similarities}}}
	$K = K - \mathbf{1}_{1/T}K - K\mathbf{1}_{1/T} + \mathbf{1}_{1/T}K\mathbf{1}_{1/T}$ \tcc*[r]{Standardization}
	$[A,\Lambda] = K/T$ \tcc*[r]{Eigendecomposition}
	$\widehat{F}_r = KA_r$ \tcc*[r]{Compute factors}
	\KwOut{$T\times r$ matrix $\widehat{F}_r$.} 
	\caption{kPCA Algorithm}\label{alg:alg1}
\end{algorithm} 

The algorithm starts by computing similarities via a given kernel.  This induces transformed observation that need to be demeaned according to equation \eqref{eq:eq5}.
As in ordinary PCA, one then needs to compute the eigenvectors. The difference here is that our object is a Gram matrix, not a covariance matrix. The kernel factors are then simply derived as projections of the kernel matrix onto $r$ eigenvectors associated with the largest eigenvalues.

\subsection{Data and Forecast Construction}
As an empirical investigation, we examine whether using kernel factors leads to improved performance in forecasting several key macroeconomic indicators. We use a large dataset from FRED-MD (\cite{McCracken}), which has become one of the classical datasets for empirical analysis of big data. Its latest release consists of 128 monthly US variables running from $1959:01$ through $2020:04$, 736 observations in total. Following previous studies, we set $1960:01$ as the first sample, leaving 724 observations. Since the models presented in this study require stationary series, each of the variables undergo a transformation to achieve stationarity. The decision on a particular form of transformation is generally dependent on the outcome of a unit root test, which is known to lack power in finite samples. So instead, following \cite{McCracken}, all interest and unemployment are assumed to be $I(1)$, while price indexes are assumed to be $I(2)$. The transformations applied to each series are described in supplemental materials. 

We aim to predict a single time series from this dataset by utilizing the remaining variables. The series to be predicted include 8 variables characterizing different aspects of the economy. Specifically, we take one series from each of the eight variable ``groups" in the dataset.
The summary is provided in Table~\hyperref[app:series]{A1} in Appendix.

Forecasts are constructed for $h=1,3,6,9,12,18,24$ months ahead with a rolling time window, the size of which is taken to be $120-h$. Thus, the pseudo-out-of-sample forecast evaluation period is $1970:01$ to $2020:04$, which is 604 months. We estimate 6 variants of autoregressive diffusion index models. The first model, taken as a benchmark, is a classical ARDI with PC estimates. Several studies have documented a strong performance of this model (see for example, \cite{Stevanovich2019}). The second and third take SPC and so-called PC-squared ($PC^2$) estimates (\cite{Bai2008}) respectively, where the latter is identical to the first model with squares of factor estimates added in the forecasting equation. The remaining models are based on  kPCA estimates with three different kernels: a sigmoid $k(\mathrm{x_i}, \mathrm{x_j}) = \tanh(\gamma(\mathrm{x_i}'\mathrm{x_j})+1)$, a radial basis function (RBF) $k(\mathrm{x_i}, \mathrm{x_j}) = e^{-\gamma\norm{\mathrm{x_i} - \mathrm{x_j}}^2}$ and a quadratic\footnote{Polynomial kernels of lower and higher order demonstrated poor forecasting ability and are not included.} polynomial (poly(2)) kernel $k(\mathrm{x_i}, \mathrm{x_j}) = (\mathrm{x_i}'\mathrm{x_j}+1)^d, \; d=2$. 

Optimal parameters for each model at each step, $P^h_t,M^h_t,K^h_t$ and the kernel hyperparameter, are determined within the rolling window period, that is our setup only permits the information set that would be available at the moment of making a prediction. Specifically, $P^h_t,M^h_t,K^h_t$ are selected by BIC (maximum value allowed for each is set equal to $6$) for each out-of-sample period, while $\gamma$ is determined over a grid of values by so-called time series cross-validation. Specifically, we consecutively predict the latest 5 available observations and select the hyperparameter that minimizes the average error. The standard cross-validation may not theoretically be fully adequate due to the presence of serial correlation in the data and several approaches were suggested to correct it (\cite{Racine00}). 

\subsection{Results}
The main empirical findings are presented in Table \ref{tab:tab1}. The subtitles of series indicate the names of dependent variables, while each value in the table represents the ratio of out-of-sample MSPE of a given estimation method to out-of-sample MSPE of the autoregression augmented by diffusion indexes estimated by PCA. The results range for $8$ variables across $7$ different forecast horizons. Our results are reproducible: in supplemental materials we provide the script written in Python $3.6$ that generates all results within a few hours.

\begin{center} 
	\setstretch{1}
	\begin{longtable}{lrrrrrrr}
		\caption{Relative MSPEs for 8 variables across 7 different prediction horizons. \vspace{10pt} \\
		\small \centering
		Each value represents the ratio to out-of-sample MSPE of the autoregression augmented by diffusion indexes estimated by PCA . An asterisk indicates the best performer for each horizon (no asterisk indicates the superiority of the baseline method). Values printed in boldface suggest statistical significance of the Diebold-Mariano test of equal predictive ability at 90\%, when the corresponding method is compared against the autoregression with PCA estimates.} 	\label{tab:tab1}\\ 
		\toprule  {} &     $h=1$ &     $h=3$ &     $h=6$ &     $h=9$ &    $h=12$ &    $h=18$ &    $h=24$ \\  \midrule \endfirsthead
		\multicolumn{8}{c}
		{{\tablename\ \thetable{} \textit{-- continued from previous page}}} \\
		\toprule  {} &     $h=1$ &     $h=3$ &     $h=6$ &     $h=9$ &    $h=12$ &    $h=18$ &    $h=24$ \\  \midrule \\  \endhead
		\multicolumn{8}{r}{{\textit{Continued on next page}}} \\ \midrule
		\endfoot
		\endlastfoot
		%PLACE TABLE FROM PYTHON HERE!
		\multicolumn{8}{c}{Real Personal Income}\\
		\cmidrule(r){1-8}
		SPC         &  \textbf{1.0327} &  \textbf{1.3195} &  1.5310 &  1.8579 &  1.3750 &   1.1647 &   1.1979 \\
		PC$^2$          &  1.0130 &  1.1736 &  1.2709 &  1.0859 &  1.2503 &  \textbf{38.1598} &  \textbf{19.4460} \\
		kPCA poly(2) &  \textbf{1.0769} &  \textbf{2.2967} &  \textbf{3.1602} &  \textbf{2.5349} &  1.9702 &   2.2061 &   1.3281 \\
		kPCA sigmoid &  1.0013 &  0.9886$^*$ &  \textbf{0.9650}$^*$ &  0.8835$^*$ &  0.8891$^*$ &   0.8377$^*$ &   \textbf{0.8688}$^*$ \\
		kPCA RBF     &  0.9995$^*$ &  0.9972 &  0.9977 &  0.9318 &  \textbf{0.9367} &   0.9269 &   \textbf{0.9890} \\
		\cmidrule(r){1-8}
		\multicolumn{8}{c}{Civilian Employment}\\
		\cmidrule(r){1-8}
		SPC         &  1.0257 &  \textbf{1.4758} &  1.6383 &  2.2547 &  1.9158 &   1.3148 &    1.3624 \\
		PC$^2$          &  0.9871$^*$ &  1.1357 &  1.2180 &  1.2806 &  \textbf{1.4784} &  \textbf{20.8259} &  \textbf{386.2726} \\
		kPCA poly(2) &  \textbf{1.4307} &  \textbf{1.6269} &  2.7371 &  \textbf{2.0439} &  1.7979 &   1.6266 &    \textbf{1.6697} \\
		kPCA sigmoid &  1.0017 &  0.9830 &  \textbf{0.9245}$^*$ &  \textbf{0.9213}$^*$ &  \textbf{0.9084}$^*$ &   \textbf{0.9171}$^*$ &    0.9138$^*$ \\
		kPCA RBF     &  1.0004 &  \textbf{0.9758}$^*$ &  \textbf{0.9448} &  \textbf{0.9317} &  0.9381 &   1.0028 &    0.9696 \\
		\cmidrule(r){1-8}
		\multicolumn{8}{c}{Housing Starts: Privately Owned}\\
		\cmidrule(r){1-8}
		SPC         &  1.0102 &  1.5529 &  2.3532 &  1.9179 &  1.4604 &   1.2191 &   \textbf{1.7948} \\
		PC$^2$          &  \textbf{1.0978} &  1.1887 &  1.7617 &  1.8405 &  \textbf{1.3927} &  \textbf{67.1838} &  \textbf{23.4097} \\
		kPCA poly(2) &  1.0684 &  1.3471 &  2.8672 &  \textbf{1.6391} &  2.2092 &   2.2018 &   2.8618 \\
		kPCA sigmoid &  0.9953 &  0.9742 &  0.9786$^*$ &  \textbf{0.9583} &  0.9576 &   \textbf{0.9453} &   0.9771 \\
		kPCA RBF     &  0.9950$^*$ &  0.9588$^*$ &  0.9990 &  0.9225$^*$ &  \textbf{0.9412}$^*$ &   \textbf{0.8944}$^*$ &   \textbf{0.9409}$^*$ \\
		\cmidrule(r){1-8}
		\multicolumn{8}{c}{Real personal consumption}\\
		\cmidrule(r){1-8}
		SPC         &  1.0790 &  1.3103 &  1.6410 &  2.0345 &  1.5595 &   1.1362 &   1.3404 \\
		PC$^2$          &  \textbf{1.0456} &  1.1577 &  1.2311 &  1.3181 &  \textbf{1.1786} &  \textbf{11.5388} &  \textbf{19.6892} \\
		kPCA poly(2) &  \textbf{1.1722} &  2.0990 &  \textbf{1.8980} &  1.8964 &  2.5677 &   2.8196 &   1.7777 \\
		kPCA sigmoid &  1.0037 &  0.9901 &  \textbf{0.9549}$^*$ &  \textbf{0.9310}$^*$ &  \textbf{0.9072}$^*$ &   \textbf{0.9131}$^*$ &   \textbf{0.9477}$^*$ \\
		kPCA RBF     &  0.9919$^*$ &  0.9894$^*$ &  \textbf{0.9609} &  \textbf{0.9363} &  \textbf{0.9596} &   \textbf{0.9622} &   0.9799 \\
		\cmidrule(r){1-8}
		\multicolumn{8}{c}{M1 Money Stock}\\
		\cmidrule(r){1-8}
		SPC         &  1.0660 &  1.5959 &  1.3715 &  \textbf{1.5438} &  1.7104 &   1.4845 &    1.3175 \\
		PC$^2$          &  \textbf{1.3453} &  1.4641 &  1.1331 &  1.1987 &  2.5199 &  \textbf{27.0868} &  \textbf{157.3719} \\
		kPCA poly(2) &  1.9656 &  3.9986 &  2.6646 &  2.6261 &  2.5258 &   2.1058 &    1.6008 \\
		kPCA sigmoid &  1.0074 &  \textbf{0.9631} &  \textbf{0.9275}$^*$ &  \textbf{0.9183}$^*$ &  \textbf{0.9311}$^*$ &   0.8977$^*$ &    \textbf{0.8295}$^*$ \\
		kPCA RBF     &  0.9892$^*$ &  0.9994$^*$ &  \textbf{0.9663} &  \textbf{0.9569} &  \textbf{0.9756} &   \textbf{0.9521} &    0.9827 \\
		\cmidrule(r){1-8}
		\multicolumn{8}{c}{Effective Federal Funds Rate}\\
		\cmidrule(r){1-8}
		SPC         &  0.9261$^*$ &  1.4816 &  3.1258 &  2.6737 &  2.4173 &  1.7712 &   2.2228 \\
		PC$^2$          &  1.2278 &  1.3466 &  1.3976 &  2.3356 &  1.2200 &  \textbf{9.0571} &  \textbf{23.7870} \\
		kPCA poly(2) &  1.3172 &  2.2581 &  3.8848 &  2.8236 &  1.3301 &  2.8415 &   1.6740 \\
		kPCA sigmoid &  0.9921 &  0.9250$^*$ &  0.8363$^*$ &  0.8503$^*$ &  0.8654$^*$ &  \textbf{0.8094}$^*$ &   0.7950$^*$ \\
		kPCA RBF     &  1.0388 &  0.9782 &  0.9544 &  \textbf{0.9119} &  \textbf{0.9619} &  0.9335 &   0.8907 \\
		\cmidrule(r){1-8}
		\multicolumn{8}{c}{CPI: All Items}\\
		\cmidrule(r){1-8}
		SPC         &  1.0445 &  2.0341 &  2.2723 &  1.3606 &  1.3226 &   \textbf{1.3522} &   \textbf{1.3168} \\
		PC$^2$          &  1.1932 &  1.6602 &  1.2048 &  1.9320 &  1.5478 &  \textbf{12.5806} &  \textbf{34.1029} \\
		kPCA poly(2) &  1.6292 &  3.0142 &  3.9154 &  1.9549 &  1.4750 &   1.3192 &   1.2334 \\
		kPCA sigmoid     &  0.9901$^*$ &  1.0189 &  \textbf{0.9418}$^*$ &  0.9659$^*$ &  \textbf{0.9540}$^*$ &   \textbf{0.9518} &   \textbf{0.9538}$^*$ \\
		kPCA RBF     &  0.9967 &  1.0596 &  1.0101 &  1.0138 &  0.9803 &   \textbf{0.9504}$^*$ &   \textbf{0.9564} \\
		\cmidrule(r){1-8}
		\multicolumn{8}{c}{S\&P 500 Index}\\
		\cmidrule(r){1-8}
		SPC         &  \textbf{1.1235} &  1.3328 &  2.0309 &  2.0885 &  \textbf{2.0879} &   1.5848 &   2.0441 \\
		PC$^2$          &  \textbf{1.1131} &  1.2793 &  1.4291 &  1.6809 &  1.9813 &  41.7948 &  \textbf{43.1356} \\
		kPCA poly(2) &  \textbf{1.6190} &  \textbf{1.7477} &  \textbf{2.5330} &  2.3499 &  2.0381 &   1.4199 &   \textbf{1.3398} \\
		kPCA sigmoid &  0.9890$^*$ &  \textbf{0.9675}$^*$ &  0.9336$^*$ &  0.8708$^*$ &  0.8634$^*$ &   0.8125$^*$ &   \textbf{0.8749}$^*$ \\
		kPCA RBF     &  1.0001 &  0.9877 &  \textbf{0.9701} &  0.9683 &  0.9450 &   0.8319 &   0.9441 \\
		\cmidrule(r){1-8}
	\bottomrule
	\end{longtable} 
	\end{center}

%%%%%%%%%%%%%%%%%%%%%%%%%%%%%%%%%%%
%%%%%%%%%%%%%%%%%%%%%%%%%%%%%%%%%%%

Some comments are in order. First, sigmoid and RBF kernel approaches do lead to improved forecasting accuracy, especially at medium- and long-term time horizons. One of these two approaches dominates others in nearly 95\% of the cases considered. Additionally, Diebold-Mariano test (\cite{Diebold2002}) of equal predictive ability suggests that only these two methods are capable of often significantly outperforming the baseline PCA autoregression across a range of variables and horizons, while never being dominated. The kernel method is least advantageous for one-step-ahead forecasting. The phenomenon that linearity is hard to beat in a very short horizon is rather well known in the literature. Luckily, as was shown in Proposition \ref{prop2}, kPCA is capable of mimicking a linear PCA by adjusting the kernel hyperparameter closer to 0, which often leads to the parity of the two methods in near-term forecasting. While the gains are not pronounced at $h=1$, they become apparent at longer horizons. Results vary across variables, but the improvement is prevailing at medium-term horizons and is uniform in one-year and longer predictions. The superiority exhibited by kPCA in many cases is remarkable for macroeconomic forecasting literature. 

Second, both SPC and PC$^2$ perform substantially worse than a simple PCA. This result contradicts to \cite{Bai2008}, but is consistent with a recent empirical comparison of \cite{Exterkate}. Besides, PC$^2$ is extremely unreliable for long-term predictions. Similar to SPC, poly(2) kernel seeks to model the second-order features of the data and, as a result, often performs on par with SPC. As pointed out by the referee, the additional squared terms might often be leading to inefficiency. 

Ultimately, note that kPCA's computational complexity is dependent on the number of time periods for estimation $120-h$, making kPCA slightly faster in this particular exercise. Most importantly, kPCA's advantage would grow in a macroeconomic setting, where ``bigger" data (i.e. larger $N$) is becoming the norm.

\section{Concluding Remarks}
In this study we have introduced a nonlinear extension of factor modeling based on the kernel method. Although our exposition mainly focused on a feature mapping $\varphi(\cdot)$ enforcing nonlinearity, it is also convenient to think of this approach as kernel smoothing in an inner product space. That is, the kernel factor estimator implicitly relies on the weighted distances between original observations. This alternative viewpoint presumes that analyzing the variation in the inner product space, rather than the original space, may be more beneficial. This idea had a profound impact on machine learning and pattern recognition fields, especially as regards to support vector machines (SVMs). By using a positive definite kernel, one can be very flexible in the original space while effectively retaining the simplicity of the linear case in the high-dimensional feature space.

We have demonstrated that constructing factor estimates nonlinearly can be beneficial for macroeconomic forecasting. Specifically, the nonlinearity induced by the sigmoid and RBF kernels leads to considerable gains at medium- and long-term time horizons. This gain in performance comes at no substantial sacrifice, the algorithm remains scalable and computationally fast.

There are several possible extensions. First, it is interesting to see how the performance would change have we pre-selected the variables (targeting) before reducing the dimensionality. As shown in \cite{Bai2008} and \cite{Bulligan2015} this generally leads to better precision. Second, the forecasting accuracy can be compared with other nonlinear dimension reduction techniques mentioned earlier, such as autoencoders. For the latter, however, one must be aware of the possibility of implicit overfitting by tuning the network architecture. This is not an issue in the current framework as there are a lot fewer parameters to specify. Third, the static factors considered here could possibly be extended to dynamic factors (\cite{Forni2000a}), by explicitly incorporating the time domain, or ``efficient" factors, by weighing observations by the inverse of the estimated variance.

%%%%%%%%%%%%%%%%%%%%%%%%%%%%%%%%%%%%%%%%%%%%%%%%%%%%%%%%%%%%%%%%%%%%%%%%%%%%%%%%%%%%%%%
% APPENDIX %%%%%%%%%%%%%%%%%%%%%%%%%%%%%%%%%%%%%%%%%%%%%%%%%%%%%%%%%%%%%%%%%%%%%%%%%%%%
%%%%%%%%%%%%%%%%%%%%%%%%%%%%%%%%%%%%%%%%%%%%%%%%%%%%%%%%%%%%%%%%%%%%%%%%%%%%%%%%%%%%%%%
\newpage
\renewcommand{\appendixpagename}{Appendix}
\renewcommand{\thesubsection}{A.\arabic{subsection}}
\begin{appendices} 
	
\subsection{Decomposition of RBF kernel} \label{app:rbf}
Let $\mathrm{x},\mathrm{z} \in \mathbb{R}^k$ and $k(\mathrm{x},\mathrm{z}) = e^{-\gamma\norm{\mathrm{x} - \mathrm{z}}^2}$. Then through the Tailor expansion we can write
\begin{align*}
	k(\mathrm{x},\mathrm{z}) &= e^{-\gamma\norm{\mathrm{x}}^2}e^{-\gamma\norm{\mathrm{z}}^2}e^{2\gamma \mathrm{x}'\mathrm{z}} \\
	& = e^{-\gamma\norm{\mathrm{x}}^2}e^{-\gamma\norm{\mathrm{z}}^2} \sum_{j=0}^{\infty} \frac{{(2\gamma)^j} }{j!} (\mathrm{x}'\mathrm{z})^j\\
	& = e^{-\gamma\norm{\mathrm{x}}^2}e^{-\gamma\norm{\mathrm{z}}^2} \sum_{j=0}^{\infty} \frac{{(2\gamma)^j} }{j!}
	\sum_{\sum_{i=1}^k n_i = j} j! \prod_{i=1}^{k} \frac{(x_iy_i)^{n_i}}{n_i!}\\
	& = \sum_{j=0}^{\infty} \sum_{\sum_{i=1}^k n_i = j} \left( (2\gamma)^{j/2} e^{-\gamma\norm{\mathrm{x}}^2} \prod_{i=1}^{k} \frac{x_i^{n_i}}{\sqrt{n_i!}} \right) \left( (2\gamma)^{j/2} e^{-\gamma\norm{\mathrm{z}}^2}    \prod_{i=1}^{k} \frac{y_i^{n_i}}{n_i!} \right)\\
	& = \varphi(\mathrm{x})'\varphi(\mathrm{z})
  \end{align*}
	
That is, $\displaystyle \varphi_j(\mathrm{x}) = \sum_{\sum_{i=1}^k n_i = j} (2\gamma)^{j/2} e^{-\gamma\norm{\mathrm{x}}^2} \prod_{i=1}^{k} \frac{x_i^{n_i}}{\sqrt{n_i!}}, \quad j=0,\ldots,\infty.$

\subsection{Proof of Proposition \ref{prop2}}\label{app:prop2}
\begin{proof}
	(a) First show that $\underset{\gamma \to 0}{\lim} (2\gamma)^{-1} K = XX'$, or $\underset{\gamma \to 0}{\lim} (2\gamma)^{-1} k_{ij} \allowbreak = X_i'X_j, \; \forall i,j=1\ldots T$, where
	$k_{ij} = \widetilde{k}_{ij} - \frac{1}{T} \sum_{l=1}^T \widetilde{k}_{il}  - \frac{1}{T} \sum_{s=1}^T \widetilde{k}_{sj} - \frac{1}{T^2} \sum_{m,p=1}^T \widetilde{k}_{mp}$ and $\widetilde{k}_{ij} = e^{-\gamma\norm{X_i - X_j}_2^2}$. By L'Hopital's rule we can write $\underset{\gamma \to 0}{\lim} (2\gamma)^{-1} k_{ij} $ as \vspace{-10pt}
	\[-\frac{1}{2}\norm{X_i - X_j}_2^2 + \frac{1}{2T}\sum_{l=1}^T \norm{X_i - X_l}_2^2 + \frac{1}{2T} \sum_{l=1}^T \norm{X_l - X_j}_2^2 - \frac{1}{2T^2} \sum_{l,m=1}^T \norm{X_l - X_m}_2^2. \vspace{-10pt}\]
	 Next, use the fact that $X$ is centered, that is $X'\mathbf{1} = \mathbf{0}$ (zero column means) and hence $\sum_{l=1}^T X_i'X_l = \sum_{l=1}^T X_l'X_i = 0, \; \forall i=1 \ldots T$. This allows to simplify the above as\\
	$\displaystyle  -\frac{1}{2} X_i'X_i + X_i'X_j -\frac{1}{2}X_j'X_j 
	+ \frac{1}{2} X_i'X_i + \frac{1}{T}\sum_{l=1}^T X_l'X_l + \frac{1}{2} X_j'X_j 
	- \frac{1}{T} \sum_{l=1}^T X_l'X_l = X_i'X_j, \newline \forall i,j=1\ldots T,$ completing the first step. Hence, since the eigenvectors are normalized, we have $\underset{\gamma \to 0}{\lim}  (2\gamma)^{-1} K eig_r(K) = sXX' eig_r(XX')$ for $s$ equal either to $+1$ or $-1$. Second, given the SVD decomposition of $X = UDV'$, we have $XX' = VD^2V'$ and $X'X = UD^2U'$, with $D^2=L$. Thus, $XX' eig_r(XX') L^{-1/2} = UD = X eig_r(X'X).$ \qed
	
	(b) First show that $\underset{\gamma \to 0}{\lim} (\gamma(1-\tanh^2(c_0))^{-1} k_{ij} = X_i'X_j, \; \forall i,j=1\ldots T$, where $k_{ij} = \widetilde{k}_{ij} - \frac{1}{T} \sum_{l=1}^T \widetilde{k}_{il}  - \frac{1}{T} \sum_{s=1}^T \widetilde{k}_{sj} - \frac{1}{T^2} \sum_{m,p=1}^T \widetilde{k}_{mp}$ and $\widetilde{k}_{ij} = \tanh(c_0+\gamma X_i'X_j)$. By L'Hopital's rule \vspace{-10pt}
	\[\underset{\gamma \to 0}{\lim} (\gamma(1-\tanh^2(c_0))^{-1} k_{ij} 
	= X_i'X_j - T^{-1} \sum_{l=1}^T X_i'X_l - T^{-1} \sum_{l=1}^T X_l'X_j + T^{-2} \sum_{l,m=1}^T X_l'X_m, \vspace{-10pt}\] 
	which immediately leads to the result once mean-zero property is taken into account. The second step is analogous to that in (a). 
\end{proof}

\subsection{Proof of Theorem \ref{thm:thm1}} \label{app:thm1}
\begin{proof}
The proof of Theorem \ref{thm:thm1} for a linear case, $\varphi(X_t) = X_t$, is available in \cite{BaiNg2002}. For a general finite-dimensional $\varphi(\cdot)$ the result follows by applying the original theorem to a vector $\varphi(X_t)$ instead.
\end{proof}

\subsection{Bounded differences inequality}
\begin{theorem*}[\cite{Mcdiarmid89}]
	Given independent random variables $X_1, \ldots, X_n \in \mathcal{X}$ and a mapping $f: \mathcal{X}^n \to \mathbb{R}$ satisfying 
	\[ \underset{x_1, \ldots, x_n, x_i' \in \mathcal{X}}{\sup}  \abs{f(x_1, \ldots, x_i, \ldots, x_n) - f(x_1, \ldots, x_i', \ldots, x_n)} \le c_i, \] then for all $\epsilon>0$,
	\[\mathbb{P}(f(X_1,\ldots,X_n) - \E(f(X_1,\ldots,X_n))\ge \epsilon) \le e^{\frac{-2\epsilon^2}{\sum_{i=1}^{n} c_i^2}}. \]
\end{theorem*}\label{thm:bdi}

\subsection{Proof of Lemma \ref{lem:lem1}} \label{app:lem1}
\begin{proof}
	Let $\Sigma_x := \varphi(x)\otimes\varphi(x)$. Note that $\norm{\Sigma_x} = k(x,x) \le \bar{k}$, and hence
	\[ \underset{x_1, \ldots, x_T, x_i' \in \mathcal{X}}{\sup} \; 
	\abs{ \norm{\frac{1}{T} \sum_{x_1\ldots x_i \ldots x_T} \Sigma_{x_i} - \E\Sigma_X} - \norm{\frac{1}{T} \sum_{x_1\ldots x_i' \ldots x_T} \Sigma_{x_i} - \E\Sigma_X} } \le \]
	\[ \underset{x_i \in \mathcal{X}}{\sup} \; \frac{1}{T} \abs{\norm{\Sigma_{x_i} - \E\Sigma_X}} \le \frac{2\bar{k}}{T}. \]
	Thus, by bounded difference inequality (\cite{Mcdiarmid89}) we have \[\mathbb{P}\left(\norm{\widehat{\Sigma}-\Sigma} -  \E\left(\norm{\widehat{\Sigma}-\Sigma}\right) \ge 2\bar{k}\sqrt{\frac{\epsilon}{2T}}\right) \le e^{-\epsilon}.\]
	Finally,
	\[\E\left(\norm{\widehat{\Sigma}-\Sigma}\right) \le \E\left(\norm{\widehat{\Sigma}-\Sigma}^2\right)^{1/2} = T^{-1/2} \E\left(\norm{\Sigma_X-\E(\Sigma_X)}^2\right)^{1/2} \le \frac{2\bar{k}}{\sqrt{T}},\]
	since $\E\left(\norm{\Sigma_X-\E(\Sigma_X)}^2\right) = \inner{\Sigma_X-\E\left(\Sigma_X\right), \Sigma_X-\E\left(\Sigma_X\right)} \le 4\bar{k}^2$.
\end{proof}

\subsection{Proof of Lemma \ref{lem:lem2}} \label{app:lem2}
\begin{proof}
	See the proof of Lemma 5.2 in \cite{Koltchinskii2000}.
\end{proof}

\subsection{Proof of Theorem \ref{thm:thm2}} \label{app:thm2}
\begin{proof}
	Since $\psi_i,\widehat{\psi}_i$ are standardized to be unit length, we have $\inner{\psi_i,\widehat{\psi_i}}^2 \le 1$ by Cauchy-Schwarz inequality. Choosing eigenvector signs so that $\inner{\psi_i, \widehat{\psi}_i} >0$, we have
	\[ \norm{\psi_i - \widehat{\psi}_i}^2 = 2 - 2\inner{\psi_i,\hat{\psi}_i} \le 2 - 2\inner{\psi_i,\widehat{\psi}_i}^2 = \norm{\Pi_i(\Sigma) - \Pi_i(\widehat{\Sigma})}^2.\]
	Using Lemma \ref{lem:lem2},
	\[\norm{\Pi_i(\Sigma) - \Pi_i(\widehat{\Sigma})} \le 4 \delta_i^{-1}(\Sigma) \norm{\widehat{\Sigma}-\Sigma}, \]
	and hence through Lemma \ref{lem:lem1} we have
	\[ \mathbb{P}\left(\norm{\widehat{\psi}_i-\psi_i} \ge \left( 1+\sqrt{\frac{\epsilon}{2}} \right) \frac{8\bar{k}}{\sqrt{T}\delta_i(\Sigma)}\right) \le e^{-\epsilon},
	\]
	which implies the result. 
\end{proof}

\subsection{Mercer's Theorem} \label{app:Mercer}
\begin{theorem*}[\cite{Mercer}]
	Given compact $\mathcal{X} \subseteq \mathbb{R}^d$ and continuous $K: \mathcal{X}\times \mathcal{X} \to \mathbb{R}$, satisfying 
	\[ \int_\mathrm{y}\int_\mathrm{x} K^2(\mathrm{x},\mathrm{y})d\mathrm{x}d\mathrm{y} < \infty \; \text{ and } \int_\mathrm{y}\int_\mathrm{x} f(\mathrm{x})K(\mathrm{x},\mathrm{y})f(\mathrm{y})d\mathrm{x}d\mathrm{y} \ge 0, \quad \forall f\in L^2(\mathcal{X}),\] 
	where $L^2(\mathcal{X}) = \{ f: \int f^2(\mathrm{x})d\mathrm{x} < \infty \}$,
	then there exist $\lambda_1 \ge \lambda_2 \ge \ldots \ge 0$ and functions $\{ \psi_i(\cdot) \in L^2(\mathcal{X}), i = 1,2, \ldots \}$ forming an orthonormal system in $L^2(\mathcal{X})$, i.e. $\displaystyle \inner{\psi_i,\psi_j}_{L^2(\mathcal{X})} = \int \psi_i(\mathrm{x}) \psi_j(\mathrm{x})d\mathrm{x} = \mathds{1}_{\{i=j\}}$,
	such that \[K(\mathrm{x},\mathrm{y}) = \sum_{i=1}^\infty \lambda_i\psi_i(\mathrm{x}) \psi_i(\mathrm{y}), \quad \forall \mathrm{x},\mathrm{y} \in \mathcal{X}.\] 
\end{theorem*}

\setcounter{table}{0}
\renewcommand{\thetable}{A\arabic{table}}
\subsection{Time Series} \label{app:series}
\begin{table}[htb!]
	\caption{Target variables from FRED-MD dataset.}
	\centering
	\begin{tabular}{lll}
	  \toprule
	  Group     & Fred-code     & Description \\
	  \midrule
	  Output \& income & \texttt{RPI}  & Real Personal Income \\
	  Labor market & \texttt{CE16OV} & Civilian Employment \\
	  Housing & \texttt{HOUST} & Housing Starts: Privately Owned \\
	  Consumption \& inventories & \texttt{DPCERA3M086SBEA} & Real personal consumption \\
	  Money \& credit & \texttt{M1SL} & M1 Money Stock \\
	  Interest \& exchange rates & \texttt{FEDFUNDS} & Effective Federal Funds Rate \\
	  Prices & \texttt{CPIAUCSL} & CPI: All Items \\
	  Stock Market & \texttt{S\&P 500} & S\&P's Common Stock Price Index \\
	  \bottomrule
	\end{tabular}
  \end{table}
\end{appendices}

%%%%%%%%%%%%%%%%%%%%%%%%%%%%%%%%%%%%%%%%%%%%%%%%%%%%%%%%%%%%%%%%%%%%%%%%%%%%%%%%%%%%%%%
% References %%%%%%%%%%%%%%%%%%%%%%%%%%%%%%%%%%%%%%%%%%%%%%%%%%%%%%%%%%%%%%%%%%%%%%%%%%
%%%%%%%%%%%%%%%%%%%%%%%%%%%%%%%%%%%%%%%%%%%%%%%%%%%%%%%%%%%%%%%%%%%%%%%%%%%%%%%%%%%%%%%
\clearpage
\spacingset{1}
\bibliographystyle{apalike} %apalike for natbib and amc w/o natbib
\bibliography{bibfile}
% \nocite{*} % only-referenced-entries toggle
\end{document}